\documentclass[twoside]{article}
\usepackage[accepted]{aistats2021}
\usepackage{mathtools,amsfonts,dsfont,amssymb,bm}
\usepackage[amsmath,thmmarks]{ntheorem}
\usepackage{subcaption}

\usepackage{microtype}
\usepackage{natbib}

\usepackage{algorithm}
\usepackage{algpseudocode}

\usepackage{enumitem}
\setlist[description]{leftmargin=0pt,labelindent=\parindent}

\DeclarePairedDelimiter{\abs}{\lvert}{\rvert}

\newcommand*\reals{\mathds{R}}
\newcommand\Sp[2]{\reals^{d^{#1}_{#2}}}
\newcommand*\MAT{\mathsf}
\newcommand*\PR{\mathds{P}}
\newcommand*\bMAT[1]{\bar{\MAT{#1}}}
\newcommand*\brMAT[1]{\breve{\MAT{#1}}}
\newcommand*\bbMAT[1]{\bar {\bar{\MAT{#1}}}}

\DeclareMathOperator\DARE{DARE}
\DeclareMathOperator\VAR{var}
\DeclareMathOperator\TR{Tr}

\newcommand*\EXP{\mathbb{E}}

\newcommand*\TRANS{{\mathpalette\doTRANS\empty}}
\makeatletter
\newcommand*\doTRANS[2]{\raisebox{\depth}{$\m@th#1\intercal$}}
\makeatother

\newcommand*\VVEC{\bm}
\newcommand*\BAR[1]{\bm{\bar{#1}}}

\DeclareMathOperator{\VEC}{vec}
\DeclareMathOperator{\DIAG}{diag}
\DeclareMathOperator{\ROWS}{rows}
\DeclareMathOperator{\COLS}{cols}

\newcommand\SEQ[1]{#1^1, \dots, #1^{|M|}}

\usepackage{xparse}
\NewDocumentCommand\AVG{s}
    {\IfBooleanTF#1%
      {\frac{1}{\abs{N  }} \sum_{i \in N  }}
      {\frac{1}{\abs{N^m}} \sum_{i \in N^m}}
    }

\usepackage{xcolor}

\begin{document}

\raggedbottom

\twocolumn[

\aistatstitle{Thompson sampling for linear quadratic mean-field teams}

\aistatsauthor{ Mukul Gagrani \And Sagar Sudhakara \And  Aditya Mahajan \And Ashutosh Nayyar \And Yi Ouyang }

\aistatsaddress{ USC \And  USC \And Mcgill \And USC \And Preferred Networks } ]

\begin{abstract}
We consider optimal control of an unknown multi-agent linear quadratic (LQ)
system where the dynamics and the cost are coupled across the agents through
the mean-field (i.e., empirical mean) of the states and controls.
Directly using single-agent LQ learning algorithms in such models results in 
regret which increases polynomially with the number of agents. We
propose a new Thompson sampling based learning algorithm which exploits the
structure of the system model and show that the expected Bayesian regret of
our proposed algorithm for a system with agents of $|M|$ different types at
time horizon $T$ is $\tilde{\mathcal{O}} \big( |M|^{1.5} \sqrt{T} \big)$ irrespective
of the total number of agents, where the $\tilde{\mathcal{O}}$ notation hides
logarithmic factors in $T$. We present detailed numerical experiments to
illustrate the salient features of the proposed algorithm.
\end{abstract}

\section{Introduction}
Linear dynamical systems with a quadratic cost (henceforth referred to as LQ
systems) are one of the most commonly used modeling framework in robotics,
aerospace, electrical circuits, mechanical systems, thermodynamical systems,
and chemical and industrial plants. Part of the appeal of LQ models is
that the optimal control action in such models is a linear or affine function
of the state; therefore, the optimal policy is easy to identify and easy to
implement.

Broadly speaking, three classes of learning algorithms have been proposed for
LQ systems: Optimism in the face of uncertainty (OFU) based algorithms,
certainty equivalence (CE) based algorithms, and Thompson sampling (TS) based
algorithms.

OFU-based algorithms are inspired by the OFU principle for multi-armed
bandits
\citep{auer2002finite}. 
Starting with the work of~\citep{campi1998adaptive,
abbasi2011regret}, most of the papers following this
approach~\citep{faradonbeh2017finite,cohen2019learning,abeille2020efficient}
provide a high probability bound on regret. As an illustrative
example, it is shown in \citep{abeille2020efficient} that, with high
probability, the regret of a OFU-based learning algorithm is $\tilde
{\mathcal{O}}( d_x^{0.5} (d_x + d_u) \sqrt{T})$, where $d_x$ is the dimension
of the state, $d_u$ is the dimension of the controls, $T$ is the time
horizon, and the $\tilde{\mathcal{O}}(\cdot)$ notation hides logarithmic
terms in~$T$.

Certainty equivalence (CE) is a classical adaptive control algorithm in
Systems and Control
\citep{astrom1994adaptive}.  
Most papers following this approach~\citep{dean2018regret,mania2019certainty,faradonbeh2020input,simchowitz2020naive} also
provide a high probability bound on regret.
As an illustrative example, it is shown in \citep{simchowitz2020naive} that,
with high probability, the regret of a CE-based algorithm is $\tilde
{\mathcal{O}}( d_x^{0.5}d_u  \sqrt{T} + d_x^2)$.

Thompson sampling (TS) based algorithms are inspired by TS 
algorithm for multi-armed bandits 
\citep{agrawal2012analysis}. 
Most papers following this
approach~\citep{ouyang2017control,ouyang2019posterior,abeille2018improved}
establish a bound on the expected Bayesian regret. As an illustrative
example, \citep{ouyang2019posterior} shows that the regret of a TS-based
algorithm is $\tilde{\mathcal{O}}(d_x^{0.5}(d_x + d_u) \sqrt{T})$. 

Two aspects of these regret bounds are important: the dependence on the time
horizon $T$ and the dependence on the dimensions $(d_x, d_u)$ of the state and
the controls. For all classes of algorithms mentioned above, the dependence on the time horizon is
$\tilde{\mathcal{O}}(\sqrt{T})$. Moreover, there are multiple papers which
show that, under different assumptions, the regret is lower bounded by 
$\Omega(\sqrt{T})$ \citep{cassel2020logarithmic,simchowitz2020naive}. So, the time
dependence in the available regret bounds is nearly order optimal. Similarly, even
though the dependence of the regret bound on the
dimensions of the state and the control varies slightly for each class
of algorithms, \cite{simchowitz2020naive} recently showed that the regret is
lower bounded by $\tilde{\Omega}(d_x^{0.5} d_u
\sqrt{T})$. So, there is only a
small scope of improvement in the dimension dependence in the regret bounds. 
{\parfillskip=0pt \emergencystretch=.5\textwidth \par}

The dependence of the regret bounds on the dimensions of the state and controls
is critical for applications such as formation control of robotic swarms and
demand response in power grids which have large numbers of agents (which can be of
the order of $10^3$ to $10^5$). In such systems, the effective
dimension of the state and the controls is $nd_x$ and $nd_u$, where $n$ is
the number of agents and $d_x$ and $d_u$ are the dimensions of the state and
controls of each agent. Therefore, if we take the regret bound of, say, 
the OFU algorithm proposed in \cite{abeille2020efficient}, the regret is
$\tilde{\mathcal{O}}( n^{1.5} d_x^{0.5} ({d_x + d_u}) \sqrt{T})$. Similar
scaling with $n$ holds for CE- and TS-based algorithms. The polynomial
dependence on the number of agents is prohibitive and, because of it, the
standard regret bounds are of limited value for large-scale systems.

There are many papers in the planning literature on the design of large-scale
systems which exploit some structural property of the system to develop
low-complexity design
algorithms~\citep{lunze1986dynamics,sundareshan1991qualitative,yang1995structural,hamilton2012patterned,arabneydi2015mft,arabneydi2016mft}.
However, there has been very little investigation on the role of such structural
properties in developing and analyzing learning algorithms.

Our main contribution is to show that by carefully exploiting the structure of
the model, it is possible to design learning algorithms for large-scale LQ
systems where the regret does not grow polynomially in the number of agents.
In particular, we investigate mean-field coupled control systems, which have
gained considerable importance in the last 10--15 years~\citep{LasryLions_2007,
HuangCainesMalhame_2007, HuangCainesMalhame_2012,
WeintraubBenkardRoy_2005,WeintraubBenkardVanRoy_2008}. There is a large
literature on different variations of such models and we refer the reader
to~\cite{gomes2014mean} for a survey. There has been
considerable interest in reinforcement learning for such
models~\citep{yang2018,subramanian2019reinforcement,tiwari2019reinforcement,
guo2019learning, subramanian2020multi, zhang2020reinforcement}, but
all of these papers focus on identifying asymptotically optimal policies and do
not characterize regret.

Our main contribution is to design a TS-based algorithm for mean-field teams
(which is a specific mean-field model proposed in~\cite{arabneydi2015mft,
arabneydi2016mft}) and show that (for a system with homogeneous agents) the
regret scales as $\tilde{\mathcal{O}}(|M|^{1.5}d_x^{0.5} (d_x + d_u)
\sqrt{T})$, where $|M|$ is the number of types.

We would like to highlight that although we focus on a TS-based algorithm in the
paper, it will be clear from the derivation that it is possible to develop
OFU- and CE-based algorithms with similar regret bounds. Thus, the main
takeaway message of our paper is that there is significant value in developing
learning algorithms which exploit the structure of the model.

\section{Background on mean-field teams}
\subsection{Mean-field teams model}
We start by describing a slight generalization of the basic model of
mean-field teams proposed in~\cite{arabneydi2015mft,
arabneydi2016mft}.

Consider a system with a large population of agents. The agents are
heterogeneous and have multiple types. 
Let $M = \{1, \dots, \abs{M}\}$ denote the set
of types of agents, $N^m$, $m \in M$, denote the set of all agents of
type~$m$, and $N = \bigcup_{m \in M} N^m$ denote the set of all agents.

\paragraph{States, actions, and their mean-fields:}
Agents of the same type have the same state and action spaces. In particular,
the state and control action of agents of type~$m$ take values in $\Sp{m}{x}$ and
$\Sp{m}{u}$, respectively. For any generic agent $i \in N^m$ of type~$m$, we
use $x^i_t \in \Sp{m}{x}$ and $u^i_t \in \Sp{m}{u}$ to denote its state and
control action at time~$t$. We use $\VVEC x_t = \VEC( (x^i_t)_{i \in
N})$ and $\VVEC u_t = \VEC( (u^i_t)_{i \in N})$ to denote the global state and
control actions of the system at time~$t$.

The empirical mean-field $(\bar x^m_t, \bar u^m_t)$ of agents of type~$m$, $m
\in M$, is defined as the empirical mean of the states and actions of all
agents of that type, i.e., 
\[
  \bar x^m_t =  \AVG x^i_t
  \quad\text{and}\quad
  \bar u^m_t = \AVG u^i_t.
\]

The empirical mean-field $(\BAR x_t, \BAR u_t)$ of the entire population is given
by
\[
  \BAR x_t = \VEC(\SEQ{\bar x_t})
  \quad\text{and}\quad
  \BAR u_t = \VEC(\SEQ{\bar u_t}).
\]

As an example, consider the temperature control of a multi-storied 
office building. In this case, $N$ represents the set
of rooms, $M$ represents the set of floors, $N^m$ represents all rooms in 
floor~$m$, $x^i_t$ represents the temperature in room~$i$, $\bar x^m_t$ represents
the average temperature in floor~$m$, and $\BAR x_t$ represents the collection of
average temperature in each floor. Similarly, $u^i_t$ represents the heat
exchanged by the air-conditioner in room~$i$, $\bar u^m_t$ represents the
average heat exchanged by the air-conditioners in floor~$m$, and $\BAR u_t$ represents
the collection of average heat exchanged in each floor of the building.
{\parfillskip=0pt \emergencystretch=.5\textwidth \par}

\paragraph{System dynamics and per-step cost:}

The system starts at a random initial state $x_1 = (x^i_1)_{i \in N}$, whose
components are independent across agents. For agent $i$ of type~$m$,
the initial state $x^i_1 \sim \mathcal N(0, \MAT X^i_1)$, and at time~$t \ge
1$, the state evolves according to 
\begin{equation}\label{eq:dynamics}
  x^i_{t+1} = \MAT A^m x^i_t + \MAT B^m u^i_t + 
  \MAT D^m \BAR x_t + \MAT E^m \BAR u_t + 
  w^i_t + v^m_t + \MAT F^m v^0_t,
\end{equation}
where $\MAT A^m$, $\MAT B^m$, $\MAT D^m$, $\MAT E^m$, $\MAT F^m$
are matrices of appropriate dimensions, $\{w^i_t\}_{t \ge 1}$,
$\{v^m_t\}_{t \ge 1}$, and $\{v^0_t\}_{t \ge 1}$ are i.i.d.\@ zero-mean Gaussian
processes which are independent of each other and the initial state. In particular, $w^i_t \in
\Sp{m}{x}$, $v^m_t \in \Sp{m}{x}$, and $v^0_t \in \Sp{0}{v}$, and $w^i_t \sim
\mathcal N(0, \MAT W^i)$, $v^m_t \sim \mathcal N(0, \MAT V^m)$, and $v^0_t
\sim \mathcal N(0, \MAT V^0)$. 

Eq.~\eqref{eq:dynamics} implies that all agents of type~$m$ have similar
dynamical couplings. The next state of agent~$i$ of type~$m$ depends on its
current local state and control action, the current mean-field of the states and
control actions of the system, and is influenced by three independent noise
processes: a local noise process $\{w^i_t\}_{t \ge 1}$, a noise process
$\{v^m_t\}_{t \ge 1}$ which is common to all agents of type~$m$, and a global
noise process~$\{v^0_t\}_{t \ge 1}$ which is common to all agents. 

At each time-step, the system incurs a quadratic cost $c(\VVEC x_t, \VVEC u_t)$ 
given by
\begin{multline}\label{eq:per-step-cost}
  c(\VVEC x_t, \VVEC u_t) = \BAR x_t^\TRANS \bMAT Q \BAR x_t 
  + \BAR u_t^\TRANS \bMAT R \BAR u_t \\
  + \sum_{m \in M} \AVG\bigl[ (x^i_t)^\TRANS \MAT Q^m x^i_t 
  + (u^i_t)^\TRANS \MAT R^m u^i_t \bigr].
\end{multline}
Thus, there is a weak coupling in the cost of the agents through the
mean-field. 

\paragraph{Admissible policies and performance criterion:}

There is a system operator who has access to the states of all
agents and control actions and chooses the control action according to a
deterministic or randomized policy
\begin{equation}
  \VVEC u_t = \pi_t(\VVEC x_{1:t}, \VVEC u_{1:t-1}).
\end{equation}

Let $\VVEC \theta = (\theta^m)_{m \in M}$, where $(\theta^m)^\TRANS = 
[\MAT A^m, \MAT B^m, \MAT D^m, \MAT E^m, \MAT F^m]$, denotes the
parameters of the system dynamics. The performance of any policy $\pi =
(\pi_1, \pi_2, \dots)$ is given by
\begin{equation} \label{eq:cost}
  J(\pi; \VVEC \theta) = \limsup_{T \to \infty} \frac {1}{T}
    \EXP\Bigl[ \sum_{t=1}^{T} c(\VVEC x_t, \VVEC u_t)  \Bigr].
\end{equation}
Let $J(\VVEC \theta)$ to denote the minimum of $J(\pi; \VVEC \theta)$ over all
policies.

We are interested in the setup where the system dynamics $\VVEC \theta$ are
unknown and there is a prior $p$ on $\VVEC \theta$. 
The Bayesian \emph{regret} of a policy $\pi$ operating for a
horizon $T$ is defined as 
\begin{equation}
  R(T; \pi) \coloneqq
  \EXP^\pi \biggl[ \sum_{t=1}^T c(\VVEC x_t, \VVEC u_t) 
  - T J(\VVEC \theta) \biggr]
\end{equation}
where the expectation is with respect to the prior on $\theta$, the noise
processes, the initial
conditions, and the potential randomizations done by the policy~$\pi$.

\subsection{Planning solution for mean-field teams} \label{sec:planning}
In this section, we summarize the planning solution of mean-field teams
presented in~\cite{arabneydi2015mft, arabneydi2016mft} for a known system
model.

Define the following matrices:
\begin{align*}
  \bMAT A &= \DIAG(\SEQ{\MAT A}) + \ROWS(\SEQ{\MAT D}), \\
  \bMAT B &= \DIAG(\SEQ{\MAT B}) + \ROWS(\SEQ{\MAT E}),
\end{align*}
and let 
$\bbMAT Q = \DIAG(\SEQ{\MAT Q}) + \bMAT Q$ and 
$\bbMAT R = \DIAG(\SEQ{\MAT R}) + \bMAT R$.

It is assumed that the system satisfies the following:
\vspace*{-\baselineskip}
\begin{description}
  \item[(A1)] $\bbMAT Q > 0$ and $\bbMAT R > 0$. Moreover, for every $m \in
    M$, $\MAT Q^m > 0$ and $\MAT R^m > 0$.
  \item[(A2)] The system $(\bMAT A, \bMAT B)$ is stabilizable.\footnote{System
      matrices $(\MAT A,\MAT B)$ are said to be stabilizable if there exists a gain
      matrix $\MAT L$ such that all eigenvalues of $\MAT A + \MAT B\MAT L$ are strictly inside
    the unit circle.} Moreover, for every $m \in M$, the system $(\MAT A^m,
    \MAT B^m)$ is stabilizable.
\end{description}

Now, consider the following $|M|+1$ discrete time algebraic Riccati equations
(DARE):\footnote{For stabilizable $(\MAT A, \MAT B)$ and $\MAT Q > 0$, $\DARE(\MAT A,\MAT
  B,\MAT Q,\MAT R)$ is the unique positive semidefinite solution of
\(
  \MAT S = \MAT A^\TRANS \MAT S \MAT A - (\MAT A^\TRANS \MAT S \MAT B)
  (\MAT R + \MAT B^\TRANS \MAT S \MAT B)^{-1}(\MAT A^\TRANS \MAT S \MAT B)
  + \MAT Q.
\)}
\begin{subequations}\label{eq:DARE}
\begin{align}
  \brMAT S^m &= \DARE(\MAT A^m, \MAT B^m \MAT Q^m, \MAT R^m), 
  \quad m \in M, \\
  \bMAT S &= \DARE(\bMAT A, \bMAT B, \bbMAT Q, \bbMAT R).
\end{align}
\end{subequations}
Moreover, define
\begin{subequations}\label{eq:gains}
\begin{align}
  \brMAT L^m &=
  -\bigl( (\MAT B^m)^\TRANS \brMAT S^m \MAT B^m + \MAT R^m \bigr)^{-1}
  (\MAT B^m)^\TRANS \brMAT S^m \MAT A^m, 
  \quad m \in M,\\
  \bMAT L &=
  -\bigl( \bMAT B^\TRANS \bMAT S \bMAT B + \bbMAT R \bigr)^{-1}
  \bMAT B^\TRANS \bMAT S \bMAT A,
\end{align}
\end{subequations}
and let $\ROWS(\SEQ{\bMAT L}) = \bMAT L$. 

Finally, define $\bar w^m_t = \AVG w^i_t$, $\BAR w_t = \VEC(\SEQ{\bar
w_t})$ and $\BAR v_t = \VEC(\SEQ{v_t})$. Let
\(
  \breve {\MAT W}^m = \AVG \VAR(w^i_t - \bar w^m_t)
\)
and
\(
  \bar {\MAT W} = \VAR(\BAR w_t) + \DIAG(\MAT V^1, \dots, \MAT V^{|M|}) + 
  \DIAG(\MAT F^1 \MAT V^0, \dots, \MAT F^{|M|} \MAT V^0).
\)
Note that since the noise processes are i.i.d., these covariances do not
depend on time. 

Now, split the state $x^i_t$ of agent~$i$ of type~$m$ into two parts: the
\emph{mean-field} state $\bar x^m_t$ and the \emph{relative} state $\breve
x^i_t = x^i_t - \bar x^m_t$. Do a similar split of the controls: $u^i_t = \bar
u^m_t + \breve u^i_t$. Since $\sum_{i \in N^m} \breve x^i_t = 0$ and $\sum_{i
\in N^m} \breve u^i_t = 0$, the per-step
cost~\eqref{eq:per-step-cost} can be written as
\begin{equation}\label{eq:cost-split}
  c(\VVEC x_t, \VVEC u_t) = \bar c(\BAR x_t, \BAR u_t) + 
  \sum_{m \in M} \AVG \breve c^m(\breve x^i_t, \breve u^i_t)
\end{equation}
where $\bar c(\BAR x_t, \BAR u_t) = \BAR x_t^\TRANS \bbMAT Q \BAR x_t + 
\BAR u_t^\TRANS \bbMAT R \BAR u_t$ and $\breve c^m(\breve x^i_t, \breve u^i_t)
= (\breve x^i_t)^\TRANS \MAT Q^m \breve x^i_t + (\breve u^i_t)^\TRANS \MAT R^m
\breve u^i_t$. Moreover, the dynamics of mean-field and the relative
components of the state are:
\begin{equation}\label{eq:mf-dynamics}
  \BAR x_{t+1} = \bMAT A \BAR x_t + \bMAT B \BAR u_t + \BAR w_t + \BAR v_t + \bMAT F
  v^0_t
\end{equation}
where $\bMAT F = \DIAG(\MAT F^1, \dots, \MAT F^{|M|})$ and for any
agent~$i$ of type~$m$,
\begin{equation}\label{eq:rel-dynamics}
  \breve x^i_t = \MAT A^m_t \breve x^i_t + \MAT B^m \breve u^i_t + \breve
  w^i_t,
\end{equation}
where $\breve w^i_t = w^i_t - \bar w^m_t$.

The result below follows from \cite[Theorem 6]{arabneydi2016mft}:
\begin{theorem}\label{thm:planning}
  Under assumptions \textup{(A1)} and \textup{(A2)}, the optimal policy for
  minimizing the cost~\eqref{eq:cost} is 
  given by
  \begin{equation}\label{eq:optimal}
    u^i_t = \brMAT L^m \breve x^i_t + \bMAT L^m \BAR x_t.
  \end{equation}
  Furthermore, the optimal performance is given by
  \begin{equation} \label{eq:performance}
    J(\bm \theta) = \sum_{m \in M} \TR( \breve {\MAT W}^m \brMAT S^m )
    + \TR(\bar {\MAT W} \bMAT S).
  \end{equation}
\end{theorem}

\paragraph{Interpretation of the planning solution:}
Note that $\BAR u_t = \bMAT L_t \BAR x_t$ is the optimal control for the
mean-field system with dynamics~\eqref{eq:mf-dynamics} and per-step cost $\bar
c(\BAR x_t, \BAR u_t)$. Moreover, for agent~$i$ of type~$m$, $\breve u^i_t =
\brMAT L^m_t \breve x^i_t$ is the optimal control for the relative system with
dynamics~\eqref{eq:rel-dynamics} and per-step cost $\breve c^m(\breve x^i_t,
\breve u^i_t)$. Theorem~\ref{thm:planning} shows that at every agent $i$ of
type~$m$, we can consider the two
decoupled systems---the mean-field system and the relative system---
solve them separately, and then simply add their respective
controls---$\bar u^m_t$ and $\breve u^i_t$---to obtain the optimal control action at
agent~$i$ in the original mean-field team system. We will exploit this feature
of the planning solution in order to develop a learning algorithm for mean-field teams.

\section{Learning for mean-field teams}
For the ease of exposition, we describe the algorithm for the special case when
all types are of the same dimension (i.e., $d^m_x = d_x$ and $d^m_u = d_u$ for all
$m \in M$) and the same number of agents (i.e., $|N^m| = n$ for all $m \in
M$). We further assume that $d^0_v = d_x$ and $\MAT F^m = I$. Moreover, we
assume noise covariances are given as $\MAT W^i = \sigma_w^2 I$, $i \in N$,
$\MAT V^m = \sigma_v^2I$, $m \in M$, and $\MAT V^0 = \sigma_{v^0}^2I$. 

The above assumptions are not strictly needed for the analysis but we impose them
because, under these assumptions, the covariance matrices $\bar \Sigma$ and
$\breve \Sigma^m$ are scaled identity matrices. In particular, for any $m \in
M$, $\breve \Sigma^m = (1 - \frac 1n) \sigma_w^2 I \eqqcolon \breve \sigma^2
I$ and $\bar \Sigma = (\frac{\sigma_w^2}{n} + \sigma_v^2 + \sigma_{v^0}^2) I
\eqqcolon \bar \sigma^2 I$. This simpler form of the covariance matrices
simplifies the description of the algorithm and the regret bounds.

Following the decomposition presented in Sec.~\ref{sec:planning}, we define
$\bar \theta^\TRANS = [\bMAT A, \bMAT B]$ to be the parameters of the
mean-field dynamics~\eqref{eq:mf-dynamics} and 
$(\breve \theta^m)^\TRANS = [\MAT A^m,
\MAT B^m]$
to be the parameters of the relative dynamics~\eqref{eq:rel-dynamics}. We let
$\brMAT S^m(\breve \theta^m)$ and $\bMAT S(\bar \theta)$ denote the
solution to the Riccati equations~\eqref{eq:DARE} and $\brMAT L^m(\breve
\theta^m)$ and $\bMAT L(\bar \theta)$ denote the corresponding
gains~\eqref{eq:gains}. Let $\breve J^m(\breve \theta^m) = \breve \sigma^2
\TR(\brMAT S(\breve \theta^m))$ and $\bar J(\bar \theta) = \bar \sigma^2
\TR(\bMAT S(\bar \theta))$ denote the performance of the $m$-th
relative system and the mean-field system, respectively. As shown in
Theorem~\ref{thm:planning},
\begin{equation}\label{eq:performance-split}
  J(\VVEC \theta) = \sum_{m \in M} \breve J^m(\breve \theta^m) + \bar J(\bar
  \theta).
\end{equation}

\paragraph{Prior and posterior beliefs:}

We assume that the unknown parameters $\breve \theta^m$, $m \in M$, lie in
compact subsets $\breve \Theta^m$ of $\reals^{(d_x + d_u)
\times d_x}$. Similarly, $\bar \theta$ lies in a compact subset $\bar \Theta$
of $\reals^{|M|(d_x + d_u) \times |M| d_x}$. 
Let $\breve \theta^m(\ell)$ denote the $\ell$-th column of
$\breve \theta^m$. Thus $\breve \theta^m = \COLS(\breve \theta^m(1), \dots, \breve
\theta^m(d_x))$. Similarly, let $\bar \theta(\ell)$ denote the $\ell$-th
column of $\bar \theta$. Thus, $\bar \theta = \COLS(\bar \theta(1), \dots,
\bar \theta(|M|d_x))$. 


We use $\mathcal{N}(\mu, \Sigma)$ to denotes the Gaussian
distribution with mean $\mu$ and covariance $\Sigma$ and $p\bigr|_{\Theta}$ to
denote the projection of probability distribution $p$ on the set $\Theta$.

We assume that the priors 
$\bar p_1$ and $\breve p_1^m, m \in M,$ on $\bar \theta$  and  $\breve \theta^m, m \in M$, respectively, 
satisfy the following:

\begin{description}
 
     \item[(A3)]  $\bar p_1$ is given as:
 \[
      \bar p_1(\bar \theta) =  \big[\prod_{\ell=1}^{|M| d_x} \bar \lambda_1^{\ell}(\bar \theta(\ell))\big]\Bigr|_{\bar \Theta}   \]
     where for $\ell \in \{1, \dots, |M|d_x\}$, $\bar \lambda_1^{\ell} =
     \mathcal{N}(\bar \mu_1(\ell), \bar \Sigma_1)$,
     $\bar \mu_1(\ell) \in \reals^{|M|(d_x + d_u)}$, and
     $\bar  \Sigma_1 \in \reals^{|M|(d_x + d_u) \times |M|(d_x + d_u)}$
     is a positive definite matrix.  
  \item[(A4)] $\breve p^m_1$ is given as:
   \[
      \breve p^m_1(\breve \theta^m) =  \big[\prod_{\ell=1}^{d_x} \breve \lambda^{m,\ell}_1(\breve \theta^m(\ell))\big]\Bigr|_{\breve \Theta^m}  \]
      where for $\ell \in \{1, \dots, d_x\}$, $\breve \lambda_1^{m,\ell} = \mathcal{N}(\breve \mu^m_1(\ell), \breve \Sigma^m_1),$
   $\breve \mu^m_1(\ell) \in
   \reals^{d_x + d_u}$,  and 
    $\breve \Sigma^m_1 \in \reals^{(d_x + d_u) \times (d_x + d_u)}$  is a positive definite matrix.
\end{description}

These assumptions are similar to the assumptions on the prior 
in the recent literature on TS for LQ systems~\citep{ouyang2017control, ouyang2019posterior}.

Following the discussion after Theorem~\ref{thm:planning}, we maintain
separate posterior distributions on $\bar \theta$ and $\breve \theta^m$, $m
\in M$. In
particular, we maintain a posterior distribution  $\bar p_t$ on $\bar \theta$
based on the mean-field state and action history as follows: for any Borel
subset $B$ of $\reals^{|M|(d_x + d_u) \times |M|d_x}$, 
\begin{equation}
   \bar p_t(B) = \PR(\bar \theta \in B \mid \BAR x_{1:t}, \BAR u_{1:t-1}). \label{eq:bar_posterior}
\end{equation}

For every $m \in M$, we also maintain a separate posterior distribution
$\breve p^m_t$ on $\breve \theta^m$ as follows. At each time $t > 1$, we select
an agent $j^m_{t-1} \in N^m$ as $\arg \max_{i \in N^m} (\breve z^i_{t-1})^\TRANS
\breve \Sigma^m_{t-1} \breve z^i_{t-1}t$, where $\breve \Sigma^m_{t-1}$ is a
covariance matrix defined recursively
by~\eqref{eq:sigma_breve_update}. Then, for any Borel subset $B$ of
$\reals^{(d_x + d_u) \times d_x}$, 
\begin{equation}
  \breve p^m_t(B) = \PR(\breve \theta^m \in B \mid 
    \{ \breve x^{j^m_s}_s, \breve u^{j^m_s}_s, \breve x^{j^m_s}_{s+1} \}_{1 \le
      s < t }\} ), \label{eq:breve_posterior}
\end{equation}
See the supplementary file for a discussion on the rule to select $j^m_{t-1}$.

For the ease of notation, we use $\BAR z_t = \VEC(\SEQ{\bar z_t})$, where $\bar
z^m_t = \VEC(\bar x^m_t, \bar u^m_t)$, and $\breve z^i_t = \VEC(\breve x^i_t,
\breve u^i_t)$. Then, we can write the
dynamics~\eqref{eq:mf-dynamics}--\eqref{eq:rel-dynamics} of the mean-field and
the relative systems~as
\begin{subequations}\label{eq:bar-dynamics}
\begin{align}
  \BAR{x}_{t+1} &= \bar{\theta}^\TRANS \BAR{z}_t + \BAR w_t + \BAR v_t + v^0_t,
  \label{eq:mf-bar}
  \\
  \breve{x}_{t+1}^i &= (\breve{\theta}^m)^\TRANS \breve{z}^i_t +
  \breve{w}_t^i,
  \quad \forall i \in N^m, m \in M.
  \label{eq:rel-bar}
\end{align}
\end{subequations}

Recall that $\bar \sigma^2 = \sigma_w^2/n + \sigma_v^2 + \sigma_{v^0}^2$ and
$\breve \sigma^2 = (1 - \frac1n)\sigma_w^2$. 

\begin{lemma}
  The posterior distributions are as follows:
  \begin{enumerate}[leftmargin=1.5em, font=\textup]
    \item The posterior on $\bar \theta$ is
     \[ \bar p_t = \big[ \prod_{\ell=1}^{|M| d_x} \bar \lambda_t^{\ell}(\bar
     \theta(\ell)) \big]\Bigr|_{\bar \Theta}, \]
     where for $\ell \in \{1, \dots, |M|d_x\}$,  $\bar \lambda_t^{\ell} =
      \mathcal{N}(\bar \mu_t(\ell), \bar \Sigma_t)$, 
      and  
      \begin{subequations}\label{eq:p-bar}
      \begin{align}
        \bar \mu_{t+1}(\ell) &= \bar \mu_{t}(\ell) + 
        \frac{ 
          \bar \Sigma_t \BAR z_t 
          \bigl( \BAR x_{t+1}(\ell) - \bar \mu_t(\ell)^\TRANS
        \BAR z_t \bigr) }
        { \bar{\sigma}^2 +   (\BAR z_t)^\TRANS \bar \Sigma_t 
        \BAR z_t }, \label{eq:mu_bar_update} \\
        \bar \Sigma_{t+1}^{-1} &= \bar \Sigma_t^{-1} 
        + \frac{1}{\bar{\sigma}^2}  \BAR z_t \BAR z_t^\TRANS.
        \label{eq:sigma_bar_update}
      \end{align}
      \end{subequations}
    \item The posterior on $\breve \theta^m$, $m \in M$, at time~$t$ is
    \[ \breve p^m_t(\breve \theta^m) = \big[ \prod_{\ell=1}^{d_x} \breve
    \lambda^{m,\ell}_t(\breve \theta^m(\ell)) \big]\Bigr|_{\breve \Theta^m},\]
      where for $\ell \in \{1, \dots, d_x\}$,  $\breve \lambda^{m,\ell}_t =
      \mathcal{N}(\breve \mu^m_t(\ell), \breve \Sigma^m_t)$,
     and 
      \begin{subequations}\label{eq:p-breve}
      \begin{align}
        \breve \mu^m_{t+1}(\ell) &= \breve \mu^m_{t}(\ell) + 
        \frac{ 
          \breve \Sigma^m_t \breve z^{j^m_t}_t 
          \bigl( \breve x^{j^m_t}_{t+1}(\ell) - \breve \mu^m_t(\ell)^\TRANS
        \breve z^{j^m_t}_t \bigr) }
        { \breve{\sigma}^2 +  (\breve z^{j^m_t}_t)^\TRANS \breve \Sigma^m_t 
        \breve z^{j^m_t}_t },
        \label{eq:mu_breve_update} \\
        (\breve \Sigma^m_{t+1})^{-1} &= (\breve \Sigma^m_t)^{-1} 
        + \frac{1}{\breve{\sigma}^2} \breve z^{j^m_t}_t (\breve z^{j^m_t}_t)^\TRANS.
        \label{eq:sigma_breve_update}
      \end{align}
      \end{subequations}
  \end{enumerate}
\end{lemma}

\begin{proof}
  Note that the dynamics of $\BAR{x}_t$ and $\breve{x}^i_t$ in
  \eqref{eq:bar-dynamics} are linear and the noises $\BAR{w}_t + \BAR v_t +
  v^0_t$ and $\breve w^i_t$ are Gaussian. Therefore, the result follows from
  standard results in Gaussian linear regression
  \citep{sternby1977consistency}.
\end{proof}

\paragraph{The Thompson sampling algorithm:}

We propose a Thompson sampling algorithm referred to as \texttt{TSDE-MF} which
is inspired by the \texttt{TSDE} (Thompson sampling with dynamic episodes) algorithm
proposed in~\cite{ouyang2017control, ouyang2019posterior} and the structure of
the optimal planning solution for the mean-field teams described in
Sec.~\ref{sec:planning}.

The \texttt{TSDE-MF} algorithm consists of a coordinator $\mathcal C$ and
$|M|+1$ \emph{actors}: a mean-field actor $\bar{\mathcal A}$ and a relative
actor $\breve{\mathcal{A}}^m$, for each $m \in M$. These
actors are described below while the whole algorithm is presented in
Algorithm~\ref{alg:tsde_mf}.
\begin{itemize}[leftmargin=1em]
  \item At each time, the coordinator $\mathcal C$ observes the current global
    state $(x^i_t)_{i \in N}$, computes the mean-field state $\BAR x_t$ and
    the relative states $(\breve x^i_t)_{i \in N}$, and sends 
    the mean-field state $\BAR x_t$ to be the mean-field actor
    $\bar{\mathcal A}$ and the relative states $\bm{\breve x}^m_t =
    (\breve x^i_t)_{i \in N^m}$ of the all the agents of type~$m$ to the
    relative actor~$\breve{\mathcal A}^m$. The mean-field actor
    $\bar{\mathcal A}$ computes the
    mean-field control $\BAR u_t$ and the relative actor $\breve {\mathcal
    A}^m$ computes the relative control $\bm{\breve u}^m_t = (\breve u^i_t)_{i
    \in N^m}$ (as per the details presented below) and sends it back to the
    coordinator~$\mathcal{C}$. The coordinator then computes and executes the
    control action $u^i_t = \bar u^m_t + \breve u^i_t$ for each agent~$i$ of
    type~$m$.

  \item The mean-field actor $\bar{\mathcal{A}}$ maintains the posterior $\bar
    p_t$ on $\bar \theta$ according to~\eqref{eq:p-bar}. The actor works in
    episodes of dynamic length. Let $\bar t_k$ and $\bar T_k$ denote the start
    and the length of episode~$k$, respectively. Episode $k$ ends if the
    determinant of covariance $\bar \Sigma_t$ falls below half of its value at
    the beginning of the episode (i.e., $\det(\bar \Sigma_t) < 0.5 \det(\bar
    \Sigma_{\bar t_k})$) or if the length of the episode is one more than the
    length of the previous episode (i.e., $t - \bar t_k > \bar T_{k-1}$).
    Thus,
    \begin{multline}
      \bar t_{k+1} = \min \bigl\{ t > t_k : \det(\bar \Sigma_t) < 0.5\det(\bar
        \Sigma_{t_k})
        \\ \text{or } t - \bar t_k > \bar T_{k-1} \bigr\}.
    \end{multline}
    At the beginning of episode~$k$, the mean-field actor $\bar{\mathcal{A}}$
    samples a parameter $\bar \theta_k$ from the posterior distribution $\bar
    p_t$. During episode~$k$, the mean-field actor $\bar{\mathcal{A}}$
    generates the mean-field controls using the samples $\bar \theta_k$,
    i.e., $\BAR u_t = \bMAT L(\bar \theta_k) \BAR x_t$.

  \item Each relative actor $\breve{\mathcal A}^m$ is similar to the
    mean-field actor. Actor $\breve{\mathcal A}^m$ maintains the posterior
    $\breve p^m$ on $\breve \theta^m$ according to~\eqref{eq:p-breve}. The
    actor works in episodes of dynamic length. The episodes of each relative
    actor $\breve{\mathcal{A}}^m$ and the mean-field actor $\bar{\mathcal{A}}$
    are separate from each other.\footnote{We use the
      episode count~$k$ as a local variable which is different for each
    actor.} Let $\breve t^m_k$ and $\breve T^m_k$ denote
    the start and length of episode~$k$, respectively. The termination
    condition for each episode is similar to that of the mean-field
    actor~$\bar{\mathcal{A}}$. In particular,
    \begin{multline}
      \breve t^m_{k+1} = \min \bigl\{ t > t^m_k : \det(\breve \Sigma^m_t) < 0.5\det(\breve
        \Sigma^m_{t^m_k})
        \\ \text{or } t - \breve t^m_k > \breve T^m_{k-1} \bigr\}.
    \end{multline}
    At the beginning of episode~$k$, the relative actor $\breve{\mathcal{A}}^m$
    samples a parameter $\breve \theta^m_k$ from the posterior distribution $\breve
    p^m_t$. During episode~$k$, the relative actor
    $\breve{\mathcal{A}}^m$ generates the relative controls using the
    sample $\breve \theta^m_k$, i.e., $\bm{\breve
    u}^m_t = (\brMAT L^m(\breve \theta^m_{k})\breve x^i_t)_{i \in N^m}$.

\end{itemize}

\begin{algorithm}[!t]
\caption{\texttt{TSDE-MF}}
\label{alg:tsde_mf}
\begin{algorithmic}[1]
  \State \textbf{initialize mean-field actor:} $\bar \Theta$, $(\bar \mu_1,
  \bar \Sigma_1)$, $\bar t_0 = 0$, $\bar T_{-1} = 0$, $k = 0$
  \State \textbf{initialize relative-actor-$m$:} $\breve \Theta^m$, $(\breve
  \mu^m_1, \breve \Sigma^m_1)$, $\breve t^m_0 = 0$, $\breve T^m_{-1} = 0$, $k = 0$
\For{$t = 1, 2, \dots $}
  \State \textit{observe} $(x^i_t)_{i \in N}$
  \State \textit{compute} $\BAR x_t$, $(\bm{\breve x}^m_t)_{m \in M}$
  \State $\BAR u_t \gets \textsc{mean-field-actor}(\BAR x_t)$
  \For{$m \in M$}
      \State $\bm{\breve u}^m_t \gets \textsc{relative-actor-$m$}(\bm{\breve
      x}^m_t)$
      \For{$i \in N^m$}
      \State \textit{agent~$i$ applies control} $u^i_t = \bar u^m_t + \breve
      u^i_t$
      \EndFor
  \EndFor
 \EndFor
\end{algorithmic}
\medskip
\begin{algorithmic}[1]
\Function{mean-field-actor}{$\BAR x_t$}
  \State \textbf{global var} $t$
  \State Update $\bar p_t$ according~\eqref{eq:p-bar}
  \If{$t - \bar t_k > \bar T_{k-1}$ or $\det(\bar \Sigma_t) < 0.5\det(\bar
\Sigma_{k})$}
    \State $T_k \gets t - \bar t_k$,
    $k \gets k + 1$,
     $\bar t_k \gets t$
    \State \textit{sample} $\bar \theta_k \sim \bar p_t$
    \State $\bMAT L \gets \bMAT L(\bar \theta_k)$
\EndIf
  \State \textbf{return} $\bMAT L \BAR x_t$
\EndFunction
\end{algorithmic}
\medskip
\begin{algorithmic}[1]
\Function{relative-actor-$m$}{$(\breve x^i_t)_{i \in N^m}$}
  \State \textbf{global var} $t$
  \State Update $\breve p^m_t$ according~\eqref{eq:p-breve}
  \If{$t - \breve t^m_k > \breve T^m_{k-1}$ or $\det(\breve \Sigma^m_t) < 0.5\det(\breve
\Sigma^m_{k})$}
    \State $T^m_k \gets t - \breve t^m_k$,
    $k \gets k + 1$,
     $\breve t^m_k \gets t$
    \State \textit{sample} $\breve \theta^m_k \sim \breve p^m_t$
    \State $\brMAT L^m \gets \brMAT L^m(\breve \theta^m_k)$
\EndIf
  \State \textbf{return} $(\brMAT L^m \breve x^i_t)_{i \in N^m}$
\EndFunction

\end{algorithmic}
\end{algorithm}

Note that the algorithm does not depend on the horizon~$T$.
A partially distributed version of the algorithm is presented in the
conclusion.

\paragraph{Regret bounds:}
We make the following assumption to ensure that the closed loop dynamics of
the mean field state and the relative states of each agent are stable. We use the
notation $\|\cdot\|$ to denote the induced norm of a matrix.

\begin{description}
  \item[(A5)] There exists $\delta \in (0, 1)$ such that
    \begin{itemize}[leftmargin=1em,topsep=0pt]
      \item for any $\bar \theta, \bar \phi \in \bar \Theta$ where $\bar
        \theta^\TRANS = [\bMAT A_{\bar \theta}, \bMAT B_{\bar \theta} ]$, we
        have
        \(
          \| {\bMAT A_{\bar \theta} + \bMAT B_{\bar \theta}
          \bMAT L(\bar \phi) } \| \leq
          \delta.
        \)
      \item for any $m \in M$, 
        $\breve \theta^m,
        \breve \phi^m \in \breve \Theta^m$, 
        where $(\breve \theta^m)^\TRANS = [\MAT A_{\breve \theta^m}, \MAT B_{\breve
        \theta^m}]$, we have
    \(
      \| \MAT A_{\breve \theta^m}+ \MAT B_{\breve \theta^m} \brMAT L(\breve
      \phi^m) \|
      \leq \delta.
    \)
\end{itemize}
\end{description}

This assumption is similar to an assumption imposed in the literature on TS
for LQ systems \citep{ouyang2019posterior}. According to Theorem~11 in \cite{simchowitz2020naive}, the assumption is satisfied if 
\begin{align*}
\bar \Theta &= \{ (\bMAT A, \bMAT B): \|\bMAT A - \bMAT A_0\| \leq \bar \epsilon,
\|\bMAT B - \bMAT B_0\| \leq \bar \epsilon \}
\\
\brMAT \Theta^m &= \{ (\brMAT A^m, \brMAT B^m): ||\brMAT A^m - \brMAT A^m_0|| \leq
  \breve \epsilon^m,
\|\brMAT B^m - \brMAT B^m_0\| \leq \breve \epsilon^m \}
\end{align*}
for stabilizable $(\bMAT A^m_0, \bMAT B^m_0)$ and $(\brMAT A^m_0, \brMAT
B^m_0)$, and small constants $\bar \epsilon$ $\breve \epsilon^m$ depending on
the choice of $(\bar A^m_0, \bar B^m_0)$ and $(\brMAT A^m_0, \brMAT B^m_0)$.
In other words, the assumption holds when the true system is in a small
neighborhood of a known nominal system, and the small neighborhood can be
learned with high probability by running some stabilizing procedure
\citep{simchowitz2020naive}.
 
The following result provides an 
upper bound on the regret of the proposed algorithm.
\begin{theorem}\label{thm:main}
  Under \textup{(A1)--(A5)}, the regret of \textup{\texttt{TSDE-MF}} is upper
  bounded as follows:
  \begin{equation*}
    R(T;\textup{\texttt{TSDE-MF}} ) \leq \tilde{\mathcal O} \bigl( (\bar \sigma^2
      |M|^{1.5} + \breve \sigma^2 |M|)
    d_x^{0.5}(d_x + d_u) \sqrt{T} \bigr).
  \end{equation*}
\end{theorem}

Recall that $\bar \sigma^2 = \sigma_w^2/n + \sigma_v^2 + \sigma_{v^0}^2$ and
$\breve \sigma^2 = (1 - \frac1n)\sigma_w^2$. So, we can say that
$R(T;\texttt{TDSE-MF}) \le \tilde{\mathcal{O}}(\bar \sigma^2 |M|^{1.5}
d_x^{0.5}(d_x + d_u) \sqrt{T})$. 
Compared with the original \texttt{TSDE} regret
$\tilde{\mathcal{O}}(n^{1.5}|M|^{1.5} \sqrt{T})$
which scales superlinear with the number of agents, the regret of the proposed
algorithm is bounded by $\tilde{\mathcal{O}}(|M|^{1.5} \sqrt{T})$ irrespective of the total number of agents.

The following special cases are of
interest:
\begin{itemize}[topsep=0pt]
  \item In the absence of common noises (i.e., $\sigma_v^2 =
    \sigma_{v^{0}}^2 = 0$), and when $n \gg |M|$, 
    $R(T;\texttt{TDSE-MF}) \le \tilde{\mathcal{O}}(\breve \sigma^2 |M| d_x^{0.5}(d_x +
    d_u) \sqrt{T})$. 
  \item For homogeneous systems (i.e., $|M| = 1$), we have
    $R(T;\texttt{TDSE-MF}) \le \tilde{\mathcal{O}}( (\bar \sigma^2 + \breve
    \sigma^2) d_x^{0.5}(d_x + d_u) \sqrt{T})$. Thus, the scaling with the
    number of agents is $\tilde{\mathcal{O}}( (1 + \frac1n) \sqrt{T})$. 
\end{itemize}

Note that these results show that in mean-field systems with common noise
regret scales as $\mathcal{O}(|M|^{1.5})$ in the number of types, while in
mean-field systems without common noise, the regret scales as
$\mathcal{O}(|M|)$. \emph{Thus, the presence of common noise
fundamentally changes the scaling of the learning algorithm.}

\section{Regret analysis}
For the ease of notation, we simply use $R(T)$ instead of $R(T;\texttt{TSDE-MF})$
in this section. Eq.~\eqref{eq:performance-split} and~\eqref{eq:cost-split}
imply that the regret may be decomposed as
  \begin{equation}\label{eq:regret-split}
    R(T) = \bar R(T) + \sum_{m \in M} \frac 1n \sum_{i \in N^m} \breve R^{i,m}(T)
  \end{equation}
  where 
  \begin{align*}
    \bar{R}(T) &:= \EXP \biggl[ \sum_{t=1}^T  \bar{c}(\BAR{x}_t,\BAR{u}_t) -
    T \bar{J}(\bar{\theta}) \biggr], \\
    \breve{R}^{i,m}(T) &:= \EXP\biggl[ \sum_{t=1}^T
    \breve{c}^m(\breve{x}_t^{i},\breve{u}_t^{i} ) - T
    \breve{J}(\breve{\theta^m})
  \biggr].
  \end{align*}

Note that $\bar R(T)$ is the regret associated with the mean-field system and $\breve
R^{i,m}(T)$ is the regret of the $i$-th relative system of type~$m$. Observe
that for the mean-field actor in our algorithm is essentially implementing the
\texttt{TSDE} algorithm of \cite{ouyang2017control, ouyang2019posterior} for
the mean-field system with dynamics~\eqref{eq:mf-dynamics} and per-step cost
$\bar c(\BAR x_t, \BAR u_t)$. This is because:
\begin{enumerate}
  \item As mentioned in the discussion after Theorem~\ref{thm:planning}, we
    can view $\BAR u_t = \bMAT L(\bar \theta) \BAR x_t$ as the optimal
    control action of the mean-field system.
  \item The posterior distribution $\bar p_t$ on $\bar \theta$ depends only
    on $(\BAR x_{1:t}, \BAR u_{1:t-1})$. 
\end{enumerate}
Thus, $\bar R(T)$ is precisely the regret of the \texttt{TSDE}
algorithm analyzed in~\cite{ouyang2019posterior}. Therefore, we have the
following.
\begin{lemma}\label{lem:bar_regret_terms}
  For the mean-field system,
  \begin{equation}\label{eq:bar_regret}
    \bar{R}(T) \leq  
    \mathcal{\tilde{O}}\bigl(\bar \sigma^2 |M|^{1.5} d_x^{0.5} (d_x+d_u) \sqrt{T} 
    \bigr).
  \end{equation}
\end{lemma}

Unfortunately, we cannot use the same argument to bound $\breve R^{i,m}(T)$. Even
though we can view $\breve u^i_t = \bMAT L^m(\breve \theta^m) \breve x^i_t$ as
the optimal control action of the LQ system with
dynamics~\eqref{eq:rel-dynamics},
the posterior $\breve p^m_t$ on $\breve \theta^m$ depends on terms other than
$(\breve x^i_{1:t}, \breve u^i_{1:t-1})$. Therefore, we cannot directly use
the results of \cite{ouyang2019posterior} to bound $\breve R^{i,m}(T)$. In the
rest of this section, we present a bound on $\breve R^{i,m}(T)$.

For the ease of notation, for any episode $k$, we use $\brMAT L^m_k$ and $\brMAT
S^m_k$ to denote $\brMAT L^m(\breve \theta^m_k)$ and $\brMAT S^m(\breve \theta^m_k)$. 
Recall that the relative value function for average cost LQ problem is
$x^\TRANS \MAT S x$, where $\MAT S$ is the solution to DARE. Therefore, at any
time~$t$, episode $k$, agent~$i$ of type~$m$, and state $\breve x^i_t \in \reals^{d_x}$,
with $\breve u^i_t = \brMAT L^m_k \breve x^i_t$ and $\breve z^i_t =
\VEC(\breve x^i_1, \breve u^i_t)$, the average cost Bellman equation is
\begin{multline*}
  \breve J^m(\breve \theta^m_k) + (\breve x^i_t)^\TRANS \brMAT S^m_k \breve x^i_t =
  \breve c^m(\breve x^i_t, \breve u^i_t)  \\ + 
  \EXP\bigl[ 
    \bigl((\breve \theta^m_k)^\TRANS \breve z^i_t + \breve w^i_t\bigr)^\TRANS
    \brMAT S^m_k \bigl((\breve \theta^m_k)^\TRANS \breve z^i_t + \breve w^i_t\bigr) 
  \bigr].
\end{multline*}
Adding and subtracting 
\(  
  \EXP[ (\breve x^i_{t+1})^\TRANS \brMAT S^m_k \breve x^i_{t+1} \mid \breve z^i_t] 
\)
and noting that $\breve x^i_{t+1} = (\breve \theta^m)^\TRANS \breve z^i_t + \breve
w^i_t$, we get that
\begin{align}
   \breve c^m & (\breve x^i_t, \breve u^i_t) =
  \breve J^m(\breve \theta^m_k) 
   + (\breve x^i_t)^\TRANS \brMAT S^m_k \breve x^i_t 
  -
  \EXP[ (\breve x^i_{t+1})^\TRANS \brMAT S^m_k \breve x^i_{t+1} | \breve z^i_t] 
  \notag \\
&+ ( (\breve \theta^m)^\TRANS \breve z^i_t)^\TRANS \brMAT S^m_k ((\breve
\theta^m)^\TRANS
  \breve z^i_t)
  - ( (\breve \theta^m_k)^\TRANS \breve z^i_t)^\TRANS \brMAT S^m_k ((\breve
  \theta^m_k)^\TRANS \breve z^i_t).
  \label{eq:bellman_breve}
\end{align}
Let $\breve K^m_T$ denote the number of episodes of the relative systems of
type~$m$ until
the horizon $T$. For each $k > \breve K^m_T$, we define $\breve t^m_k$ to be $T+1$. Then, using~\eqref{eq:bellman_breve}, we have that for any
agent~$i$ of type~$m$,
\begin{align}
  \breve R^{i,m}&(T) =
  \underbrace{%
    \EXP\biggl[ \sum_{k=1}^{\breve K^m_T} \breve T^m_k \breve J^m(\breve \theta^m_k) - T 
    \breve J^m(\breve \theta^m) \biggr]
  }_{\text{regret due to sampling error\,} \eqqcolon \breve R^{i,m}_0(T)}
  \notag\\
  & +
  \underbrace{%
   \EXP\biggl[ 
    \sum_{k=1}^{\breve K^m_T} \sum_{t = \breve t^m_k}^{\breve t^m_{k+1} - 1}
    \bigl[ (\breve x^i_t)^\TRANS \brMAT S^m_k \breve x^i_t -
    (\breve x^i_{t+1})^\TRANS \brMAT S^m_k \breve x^i_{t+1} \bigr]
   \biggr]
   }_{\text{regret due to time-varying controller\,} \eqqcolon \breve R^{i,m}_1(T)}
  \notag \\
  & + 
  \underbrace{%
    \begin{aligned}[t]
   \EXP\biggl[ 
    \sum_{k=1}^{\breve K^m_T} \sum_{t = \breve t^m_k}^{\breve t^m_{k+1} - 1}
    \bigl[ ( (\breve \theta^m)^\TRANS \breve z^i_t)^\TRANS & \brMAT S^m_k ((\breve
      \theta^m)^\TRANS \breve z^i_t) 
      \\[-15pt] -
   &( (\breve \theta^m_k)^\TRANS \breve z^i_t)^\TRANS \brMAT S^m_k ((\breve
 \theta^m_k)^\TRANS \breve z^i_t) \bigr]
  \biggr].
   \end{aligned}}_{\text{regret due to model mismatch\,} \eqqcolon \breve R^{i,m}_2(T)}
   \label{eq:regret-components}
\end{align}

\begin{lemma}\label{lem:breve_regret_terms}
  The terms in~\eqref{eq:regret-components} are bounded as follows:
  \begin{enumerate}[leftmargin=1em, font=\textup]
    \item $\breve R^{i,m}_0(T) \le 
      \tilde{\mathcal{O}} (\breve \sigma^2 \sqrt{(d_x + d_u) T})$.
    \item $\breve R^{i,m}_1(T) \le 
      \tilde{\mathcal{O}} (\breve \sigma^2 \sqrt{(d_x + d_u) T})$.
    \item $\breve R^{i,m}_2(T) \le 
      \tilde{\mathcal{O}} (\breve \sigma^2 (d_x + d_u) \sqrt{d_x T})$.
  \end{enumerate}
\end{lemma}

\begin{figure*}[!htb]
    \centering
    \begin{subfigure}[t]{0.32\linewidth}
      \centering
      \includegraphics[width=\textwidth]{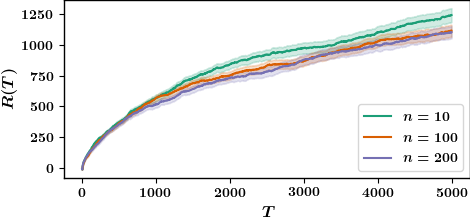}
      \caption{$R(T)$ vs $T$ for \texttt{TSDE-MF}}
      \label{fig:regreta}
    \end{subfigure}%
    \hfill
    \begin{subfigure}[t]{0.32\linewidth}
      \centering
      \includegraphics[width=\textwidth]{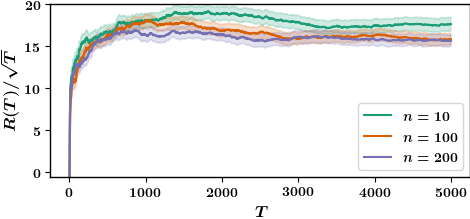}
      \caption{$R(T)/\sqrt{T}$ vs $T$ for \texttt{TSDE-MF}}
      \label{fig:regretb}
    \end{subfigure}%
    \hfill
    \begin{subfigure}[t]{0.32\linewidth}
      \centering
      \includegraphics[width=\textwidth]{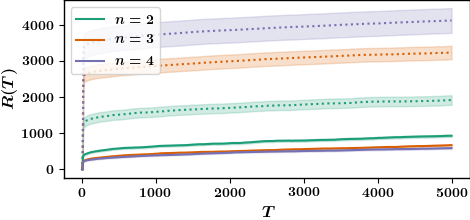}
      \caption{\texttt{TSDE-MF} (solid) vs \texttt{TSDE} (dashed)}
      \label{fig:regretc}
    \end{subfigure}%
    \caption{Expected regret vs time.} 
\end{figure*}

\begin{proof}
  We provide an outline of the proof. See the supplementary file for complete
  details.

  The first term $\breve R^{i,m}_0(T)$ can be bounded using the basic property of
  Thompson sampling: for any measurable function $f$, $\EXP[ f(\breve
  \theta^m_k) ] = \EXP[ f(\breve \theta^m)]$ because $\breve \theta^m_k$ is a sample
  from the posterior distribution on $\breve \theta^m$.  

  Note that the second term $\breve R^{i,m}_1(T)$ is a telescopic sum, which we
  can simplify to establish
  \[
    \breve R^{i,m}_1(T) \le \mathcal{O}\bigl( \EXP[ \breve K^m_T (\breve X^i_T)^2 ]\bigr),
  \]
  where $\breve X^i_t = \max_{1 \le t \le T} \| \breve x^i_t \|$ is the
  maximum norm of the relative state along the entire trajectory. The final
  bound on $\breve R^{i,m}_1(T)$ can be obtained by bounding $\breve K^m_T$
  and $\EXP[ (\breve X^i_T)^2 ]$.

  Using the sampling condition for $\breve p^m_t$ and an existing bound in the
  literature, we first establish that 
  \[
    \breve{R}_2^i(T) \leq  \sqrt{\textstyle \EXP \big[(\breve{X}_T^i)^2
    \sum_{t=1}^T (\breve{z}_t^i)^\intercal \breve{\Sigma}^m_t \breve{z}_t^i \big]}  \times \tilde{\mathcal{O}}(\sqrt{T})
\]
Then, we upper bound $(\breve{z}_t^i)^\intercal \breve{\Sigma}^m_t
\breve{z}_t^i$ by $(\breve{z}_t^{j^m_t})^\intercal \breve{\Sigma}^m_t
\breve{z}_t^{j^m_t}$ which follows from the definition of $j^m_t$. Finally, we
show that $ \EXP[(\breve{X}_T^i)^2 \sum_t (\breve{z}_t^{j^m_t})^\intercal
\breve{\Sigma}^m_t \breve{z}_t^{j^m_t}]$ is $\tilde{\mathcal{O}}(1)$ using the
fact that $(\breve{\Sigma}^m_t)^{-1}$ is obtained by linearly combining $\bigl\{
\breve z_s^{j^m_s}\bigl(\breve z_s^{j^m_s}\bigr)^\TRANS \bigr\}_{1\leq s <t}$ as in \eqref{eq:sigma_breve_update}. 
\end{proof}

Combining the three bounds in Lemma~\ref{lem:breve_regret_terms}, we get that
\begin{equation}\label{eq:breve_regret}
  \breve R^{i,m}(T) \le \tilde{\mathcal{O}}(\breve \sigma^2 d_x^{0.5}(d_x +
  d_u) \sqrt{T}).
\end{equation}
By subsituting~\eqref{eq:bar_regret} and~\eqref{eq:breve_regret}
in~\eqref{eq:regret-split}, we get the result of Theorem~\ref{thm:main}.

\section{Numerical Experiments}\label{sec:simulation}
In this section, we illustrate the performance of \texttt{TSDE-MF} for a
homogeneous (i.e., $|M| = 1$) mean-field LQ system for different values of
the number $n$ of agents. See supplementary file for the choice of parameters.

\paragraph{Empirical evaluation of regret:}

We run the system for $500$ different sample paths and plot the mean and
standard deviation of the expected regret $R(T)$ for $T = 5000$. The regret
for different values of $n$ is shown in~\ref{fig:regreta}--\ref{fig:regretb}. As seen from the
plots, the regret reduces with the number of agents and $R(T)/\sqrt{T}$
converges to a constant. Thus, the empirical regret matches the upper bound of
$\tilde{\mathcal{O}}( (1 + \frac1n)\sqrt{T})$ obtained in Theorem~\ref{thm:main}. 

\paragraph{Comparison with naive TSDE algorithm:}
We compare the performance of \texttt{TSDE-MF} with that of directly using the \texttt{TSDE}
algorithm presented in \cite{ouyang2017control,ouyang2019posterior} for
different values of $n$. The results are shown in Fig.~\ref{fig:regretc}. As
seen from the plots, the regret of \texttt{TSDE-MF} is smaller than
\texttt{TSDE} but more importantly, the regret of \texttt{TSDE-MF} reduces
with $n$ while that of \texttt{TSDE} increases with $n$. This matches their
respective upper bounds of $\tilde{\mathcal{O}}( (1 + \frac1n) \sqrt{T})$ and
$\tilde{\mathcal{O}}(n^{1.5} \sqrt{T})$. These plots clearly illustrate the
significance of our results even for small values of~$n$.

\section{Conclusion}
We consider the problem of controlling an unknown LQ mean-field
team. The planning solution (i.e., when the model is known) for
mean-field teams is obtained by solving the mean-field system and the relative
systems separately. 
Inspired by this feature, we
propose a TS-based learning algorithm \texttt{TSDE-MF} which separately tracks
the parameters $\bar \theta$ and $\breve \theta^m$ of the mean-field and
the relative systems, respectively. The part of the \texttt{TSDE-MF} algorithm
that learns the mean-field system is similar to the \texttt{TSDE}
algorithm for single agent LQ systems proposed in \cite{ouyang2017control,
ouyang2019posterior} and its regret can be
bounded using the results of \cite{ouyang2017control, ouyang2019posterior}.
However, the part of the \texttt{TSDE-MF} algorithm that learns the
relative component is different and we cannot directly use the results of
\cite{ouyang2017control, ouyang2019posterior} to bound its regret. Our main
technical contribution is to provide a bound on the regret on the relative
system, which allows us to bound the total regret under \texttt{TSDE-MF}. 

\paragraph{Distributed implementation of the algorithm:}
It is possible to implement Algorithm~\ref{alg:tsde_mf} in a distributed
manner as follows. Instead of a centralized coordinator which collects all the
observations and computes all the controls, we can consider an alternative
implementation in which there is an actor $\mathcal A^m$ associated with
type~$m$ and a mean-field actor $\bar{\mathcal{A}}$. Each agent observes its
local state and action. The actor $\mathcal{A}^m$ for type~$m$ computes
$(j^m_t, \bar x^m_t)$ using a distributed algorithm, sends $\bar x^m_t$ to the
mean-field actor, and locally computes $\brMAT L^m(\breve \theta_k)$. The
mean-field actor computes $\bMAT L(\bar \theta_k)$ and sends the $m$-th block
column $\bMAT L^m(\bar \theta_k)$ to actors $\mathcal{A}^m$. Each actor
$\mathcal{A}^m$ then sends $(\bar x^m_t, \bMAT L^m(\bar \theta_k), \brMAT
L^m(\breve \theta_k))$ to each agent of type~$m$ using a distributed
algorithm. Each agent then applies the control law~\eqref{eq:optimal}.


\acknowledgments{The work of AM was supported in part by  the Innovation for
Defence Excellence and Security (IDEaS) Program of the Canadian Department of
National Defence through grant CFPMN2-30.}

\clearpage

\bibliographystyle{myabbrvnat}
\bibliography{IEEEabrv,ref_learning}

\clearpage

\onecolumn
\appendix

\section{Appendix: Regret Analysis}

\subsection{Preliminary Results}

The analysis does not depend on the type $m$ of the agent. So for simplicity, we will omit the superscript $m$ in all the proofs in the appendix.

Since $\breve{S}(\cdot)$ and $\breve{L}(\cdot)$ are continuous functions on a compact set $\breve{\Theta}$, there exist finite constants  $\breve{M}_J, \breve{M}_{\breve{\theta}}, \breve{M}_S, \breve{M}_L$ such that $\TR(\breve{S}(\breve{\theta})) \leq \breve{M}_J, \|\breve{\theta} \| \leq \breve{M}_{\breve{\theta}}, \|\breve{S}(\breve{\theta}) \| \leq \breve{M}_S$ and $\|[I, \breve{L}(\breve{\theta})^\TRANS] \| \leq \breve{M}_L$ for all $\breve{\theta} \in \breve{\Theta}$ where $\|\cdot \|$ is the induced matrix norm.

Let
$\breve X^i_T = \max_{1 \leq t \leq T} \| \breve x^i_t \|$ be the maximum norm
of the relative state along the entire trajectory. The next bound follows
from \cite{ouyang2019posterior}[Lemma 2].

\begin{lemma}\label{lem:state_bound}
  For any $q\geq 1$ and any $T$ we have

  \begin{align*}
  \EXP \Big[ (\breve{X}_T^i)^q \Big] &\leq \breve{\sigma} ^{q}
    \mathcal{O}\Big( \log(T)(1-\delta)^{-q}\Big)
  \end{align*}
  where  $\delta$ is as defined in \textup{(A5)}.
\end{lemma}

The following lemma gives an almost sure upper bound on the number of episodes $\breve K_T$. 

\begin{lemma}\label{lem:episodes}
The number of episodes $\breve{K}_T$ is bounded as follows:

\begin{align*}
\breve{K}_T &\leq  \mathcal{O} \left( \sqrt{ (d_x+d_u) T\log \left( \frac{1}{\breve{\sigma}^2}  \sum_t (\breve{X}_T^{j_t})^2 \right)} \right)
\end{align*}
\end{lemma}

\begin{proof}
We can follow the same sketch as in proof of Lemma 3 in \cite{ouyang2019posterior}. Let $\breve{\eta} -1$ be the number of times the second stopping criterion is triggered for $\breve{p}_t$. Using the analysis in the proof of Lemma 3 in \cite{ouyang2019posterior}, we can get the following 
\begin{equation}\label{eq:macro_ep_bound}
\breve{K}_T \leq \sqrt{2 \breve{\eta} T}.
\end{equation}

Since the second stopping criterion is triggered whenever the determinant of sample covariance is halved, we have
\begin{align*}
&\det(\breve{\Sigma}^{-1}_{T}) \geq 2^{\breve{\eta}-1}\det(\breve{\Sigma}^{-1}_1)
\end{align*}

Let $d= d_x+d_u$. Since $(\frac{1}{d}\TR(\breve{\Sigma}_T^{-1}))^{d} \geq \det(\breve{\Sigma}_T^{-1})$, we have
\begin{align*}
\TR(\breve{\Sigma}_T^{-1}) & \geq d (\det(\breve{\Sigma}_T^{-1}))^{1/d}  \geq d \times 2^{(\breve{\eta}-1)/d} (\det(\breve{\Sigma}^{-1}_{1}))^{1/d} \\
&\geq d \times 2^{(\breve{\eta}-1)/d} \breve{\lambda}_{min}
\end{align*}
where $\breve{\lambda}_{min}$ is the minimum eigenvalue of $\breve{\Sigma}^{-1}_1$. 

Using \eqref{eq:sigma_breve_update} we have,
\begin{align*}
\breve{\Sigma}_{T}^{-1} = \breve{\Sigma}_1^{-1}  + \sum_{t=1}^{T-1} \frac{1}{\breve{\sigma}^2} \breve{z}_t^{j_t}  (\breve{z}_t^{j_t})^\intercal .
\end{align*}
Therefore $\TR (\breve{\Sigma}_{T}^{-1} ) = \TR(\breve{\Sigma}_1^{-1}) +  \sum_{t=1}^{T-1} \frac{1}{\breve{\sigma}^2} \TR(\breve{z}_t^{j_t}  (\breve{z}_t^{j_t})^\intercal )$. Note that $ \TR(\breve{z}_t^{j_t}  (\breve{z}_t^{j_t})^\intercal) =  \TR((\breve{z}_t^{j_t})^\intercal  \breve{z}_t^{j_t}) = \|  \breve{z}_t^{j_t}\|^2$. Thus,
\begin{align*}
d \times 2^{(\breve{\eta}-1)/d} \breve\lambda_{min} &\leq \TR(\breve{\Sigma}_1^{-1}) + \sum_{t=1}^{T-1} \frac{1}{\breve{\sigma}^2} ||\breve{z}_t^{j_t} ||^2
\end{align*}

Then,  
\begin{align*}
\breve{\eta} \leq & 1+ \frac{d}{\log2}\log\left(\frac{1}{d \breve\lambda_{min}} \left( \TR(\breve{\Sigma}_1^{-1} + \sum_{t=1}^{T-1} \frac{1}{\breve{\sigma}^2} ||\breve{z}_t^{j_t} ||^2 \right)  \right) \\
&= \mathcal{O}\left( d\log \left(\frac{1}{\breve{\sigma}^2} \sum_{t=1}^{T-1}  ||\breve{z}_t^{j_t} ||^2 \right)  \right).
\end{align*}

Note that, $||\breve{z}_t^{j_t}|| = ||[I, L(\breve{\theta})^\intercal]^\intercal \breve{x}_t^{j_t}|| \leq \breve{M}_L ||\breve{x}_t^{j_t}|| \leq \breve{M}_L \breve{X}_T^{j_t}$. Consequently,
\begin{align*}
\breve{\eta} \leq& \mathcal{O}\Big(d \log \Big(  \frac{1}{\breve{\sigma}^2} \sum_{t=1}^{T-1}  (\breve{X}_T^{j_t})^2 \Big) \Big)
\end{align*}

Therefore, combining the above inequality with \eqref{eq:macro_ep_bound} we get,
\begin{equation}
\breve{K}_T \leq \mathcal{O} \left( \sqrt{(d_x+d_u) T  \log \Big(  \frac{1}{\breve{\sigma}^2}    \sum_{t=1}^{T-1}  (\breve{X}_T^{j_t})^2 \Big) } \right)
\end{equation}

\end{proof}

\subsection{Proof of Lemma \ref{lem:breve_regret_terms}}

\begin{proof}

We will bound each part separately. 

1) Bounding $\breve{R}^i_0(T)$: From monotone convergence theorem, we have
\begin{align}
\breve{R}^i_0(T) 
= &\EXP\Big[ \sum_{k=1}^{\infty}\mathds{1}_{\{\breve t_{k}\leq T\}} \breve T_k \breve J(\breve\theta_k) \Big]- T \EXP\Big[ \breve J(\breve\theta) \Big]
\notag\\
=& \sum_{k=1}^{\infty}\EXP\Big[ \mathds{1}_{\{ \breve t_{k}\leq T\}} \breve T_k \breve J(\breve\theta_{k})\Big]  - T\EXP\Big[\breve J(\breve\theta) \Big]. \notag
\end{align}

Note that the first stopping criterion of {\tt TSDE-MF} ensures that $\breve T_k \leq \breve T_{k-1}+1$ for all $k$. Since $\breve J(\breve\theta_{k}) \geq 0$, each term in the first summation satisfies,
\begin{align*}
\EXP\Big[ \mathds{1}_{\{ \breve t_{k}\leq T\}} \breve T_k \breve J(\breve\theta_k) \Big]
\leq &\EXP\Big[ \mathds{1}_{\{\breve t_{k}\leq T\}}(\breve T_{k-1}+1) \breve J(\breve\theta_k) \Big].
\end{align*}
Note that $\mathds{1}_{\{\breve t_{k}\leq T\}}(\breve T_{k-1}+1)$ is measurable with respect to $\sigma( \{ \breve{x}_{s}^{j_s}, \breve{u}_{s}^{j_s}, \breve{x}_{s+1}^{j_s} \}_{1\leq s < \breve{t}_k } )$. Then, Lemma 4 of \cite{ouyang2019posterior} gives
\begin{align*}
\EXP\Big[ \mathds{1}_{\{\breve t_{k}\leq T\}}(\breve T_{k-1}+1) \breve J(\breve\theta_k) \Big]
= \EXP\Big[\mathds{1}_{\{\breve t_{k}\leq T\}}(\breve T_{k-1}+1) \breve J(\breve\theta) \Big].
\end{align*}

Combining the above equations, we get
\begin{align*}
\breve{R}^i_0(T)
\leq &\sum_{k=1}^{\infty}\EXP\Big[\mathds{1}_{\{\breve t_{k}\leq T\}}(\breve T_{k-1}+1) \breve J(\breve\theta) \Big] - T\EXP\Big[\breve J(\breve\theta) \Big]
\notag\\
= &\EXP\Big[\sum_{k=1}^{\breve K_T}(\breve T_{k-1}+1) \breve J(\breve\theta) \Big] - T\EXP\Big[\breve J(\breve\theta) \Big] \notag\\
= &\EXP\Big[\breve K_T \breve J(\breve\theta) \Big]
+ \EXP\Big[\Big(\sum_{k=1}^{\breve K_T}\breve T_{k-1} - T\Big)\breve J(\breve\theta) \Big] \\
\leq &  \breve M_J \breve{\sigma}^2 \EXP\Big[\breve K_T\Big] \notag
\end{align*}
where the last equality holds because $\breve J(\breve\theta) = \breve{\sigma}^2 \TR(\breve S(\breve\theta)) \leq  \breve{\sigma}^2  \breve M_J$ and $\sum_{k=1}^{\breve K_T} \breve T_{k-1} \leq T$. 

2) Bounding $\breve{R}^i_1(T)$: 

\begin{align*}
\breve{R}^i_1(T) = &\EXP\Big[\sum_{k=1}^{\breve K_T}\sum_{t=\breve t_k}^{\breve t_{k+1}-1}  \Big[ (\breve x_t^i)^\intercal \breve S(\breve\theta_k) \breve x_t^i - (\breve x_{t+1}^i)^\intercal \breve S(\breve\theta_k) \breve x_{t+1}^i\Big]\Big] \\
=&\EXP\Big[\sum_{k=1}^{\breve K_T}\Big[ (\breve x^i_{\breve t_k})^\intercal \breve S(\breve\theta_k) \breve x_{\breve t_k}^i -(\breve x_{\breve t_{k+1}}^i )^\intercal \breve S(\breve\theta_k) \breve x_{\breve t_{k+1}}^i \Big]\Big] \\
\leq &\EXP\Big[\sum_{k=1}^{\breve K_T}
(\breve x^i_{\breve t_k})^\intercal \breve S(\breve\theta_k) \breve x_{\breve t_k}^i \Big].
\end{align*}
Since $||\breve S(\breve\theta_{k})|| \leq \breve M_S$, we obtain
\begin{align*}
\breve R^i_1(T) \leq &\EXP\Big[\sum_{k=1}^{\breve K_T} \breve M_S \| \breve x_{\breve t_k}^i \|^2 \Big]
\leq \breve M_S\EXP\Big[\breve K_T (\breve X_T^i)^2 \Big].
\end{align*}

Now, from Lemma \ref{lem:episodes}, $\breve K_T \leq \mathcal{O}(\sqrt{ T \log(\frac{ \sum_{t=1}^T (\breve{X}_T^{j_t})^2}{\breve{\sigma}^2})) }$. Thus, we have $\breve{R}^i_1(T) \leq \mathcal{O} \left(\sqrt{T} \EXP\Big[ (\breve{X}_T^i)^2  \sqrt{\log \left(\frac{\sum_{t=1}^T (\breve{X}_T^{j_t})^2}{\breve{\sigma}^2} \right)} \Big] \right)$. Then, using Cauchy-Schwarz we have,

\begin{align*}
 \EXP\left[ (\breve{X}_T^i)^2  \sqrt{\log \left(\frac{\sum_{t=1}^T  (\breve{X}_T^{j_t})^2}{\breve{\sigma}^2} \right)} \right] &\leq \sqrt{\EXP \Big[(\breve{X}_T^i)^4\Big] ~~ \EXP \Big[ \log \left(\frac{ \sum_{t=1}^T (\breve{X}_T^{j_t})^2}{\breve{\sigma}^2} \right) \Big] }\\
 &\leq  \sqrt{\EXP \Big[ (\breve{X}_T^i)^4 \Big]  \log \left(\sum_{t=1}^T \frac{\EXP (\breve{X}_T^{j_t})^2}{\breve{\sigma}^2} \right) } 
\leq \tilde{\mathcal{O}}(\breve{\sigma}^2) 
\end{align*}

where the last inequality follows from Lemma \ref{lem:state_bound}. Therefore, we have  $\breve{R}^i_1(T) \leq  \tilde{\mathcal{O}} \left(\breve{\sigma}^2 ~\sqrt{T} \right)$.

3) Bounding $\breve{R}^i_2 (T)$: Each term inside the expectation of $\breve{R}^i_2$ is equal to
\begin{align*}
 \|\breve S^{0.5}(\breve\theta_k) \breve\theta^\intercal \breve z_t^i \|^2-\|\breve S^{0.5}(\breve\theta_k) \breve\theta_k^\intercal \breve z_t^i \|^2 
\leq& \left( \|\breve S^{0.5}(\breve\theta_k) \breve\theta^\intercal \breve z_t^i \| +  \|\breve S^{0.5}(\breve\theta_k) \breve\theta_k^\intercal \breve z_t^i \| \right)
||\breve S^{0.5}(\breve\theta_k) (\breve\theta- \breve\theta_k)^\intercal \breve z_t^i || \\
\leq & 2 \breve M_S \breve M_\theta \breve M_L \breve X_T^i ||(\breve\theta- \breve\theta_k )^\intercal \breve z_t^i ||
\end{align*}

since $\| \breve S^{0.5}(\breve\theta_k) \breve\phi^\intercal \breve z_t^i \| \leq \breve M_S^{0.5} \breve M_\theta \breve M_L \breve X_T^i$ for $\breve\phi = \breve\theta$ or $\breve\phi = \breve\theta_k$. Therefore,
\begin{align}
\breve R^i_2(T) \leq & 2 \breve M_S \breve M_\theta \breve M_L \EXP \Big[ \breve X_T^i \sum_{k=1}^{\breve K_T}\sum_{t=\breve t_k}^{\breve t_{k+1}-1}  \| (\breve\theta- \breve\theta_k )^\intercal \breve z_t^i \| \Big].
\label{eq:boundR2_1_breve}
\end{align}
From Cauchy-Schwarz inequality, we have
\begin{align}
\EXP\Big[\breve X_T^i \sum_{k=1}^{\breve K_T}\sum_{t=\breve t_k}^{\breve t_{k+1}-1}  \|(\breve\theta- \breve\theta_k )^\intercal \breve z_t^i \| \Big]
&= \EXP\Big[\breve X_T^i \sum_{k=1}^{\breve K_T}\sum_{t=\breve t_k}^{\breve t_{k+1}-1}  \| (\breve\Sigma^{-0.5}_t(\breve\theta-\breve\theta_k ))^\intercal  \breve\Sigma^{0.5}_t \breve z_t^i \| \Big]
\notag\\
&\leq \EXP\Big[\sum_{k=1}^{\breve K_T}\sum_{t=\breve t_k}^{\breve t_{k+1}-1}  \| \breve\Sigma^{-0.5}_t (\breve\theta- \breve\theta_k )\| \times  \breve X_T^i \|\breve\Sigma^{0.5}_t \breve z_t^i \| \Big]
\notag\\
&\leq 
\sqrt{\EXP\Big[\sum_{k=1}^{\breve K_T}\sum_{t=\breve t_k}^{\breve t_{k+1}-1}  \| \breve\Sigma^{-0.5}_t (\breve\theta- \breve\theta_k ) \|^2 \Big]}
\sqrt{\EXP\Big[\sum_{k=1}^{\breve K_T}\sum_{t=\breve t_k}^{\breve t_{k+1}-1}  (\breve X_T^i)^2 \| \breve\Sigma^{0.5}_t \breve z_t^i \|^2\Big]}
\label{eq:boundR2_2_breve}
\end{align}

From Lemma 10 in \cite{ouyang2019posterior}, the first part of \eqref{eq:boundR2_2_breve} is bounded by
\begin{align}
&\EXP\Big[\sum_{k=1}^{\breve K_T}\sum_{t=\breve t_k}^{\breve t_{k+1}-1}  \| \breve\Sigma^{-0.5}_t (\breve\theta-\breve\theta_k ) \|^2\Big]
\leq  4d_x d (T + \EXP[\breve K_T]).
\label{eq:boundR2_part1_breve}
\end{align}

For the second part of the bound in \eqref{eq:boundR2_2_breve}, we note that
\begin{align}
 \sum_t  ||\breve{\Sigma}^{0.5}_t \breve{z}_t^i ||^2 &=  \sum_{t=1}^T  ( \breve{z}_t^i)^\intercal \breve{\Sigma}_t  \breve{z}_t^i \notag \\
 &\leq \sum_{t=1}^T \max \left(1,\frac{\breve{M}_L^2 (\breve{X}_T^i)^2}{\breve{\lambda}_{min}} \right) \min(1,( \breve{z}_t^i)^\intercal \breve{\Sigma}_t  \breve{z}_t^i) \notag \\
 &\leq  \sum_{t=1}^T  \left(1 + \frac{\breve{M}_L^2 (\breve{X}_T^i)^2}{\breve{\lambda}_{min}} \right) \min(1, ( \breve{z}_t^{j_t})^\intercal \breve{\Sigma}_t  \breve{z}_t^{j_t}) \label{eq:breve_r2_4_bound}
\end{align}

where the last inequality follows from the definition of $j_t$. Using Lemma 8 of  \cite{abbasi2015bayesian} we have 

\begin{align}
\sum_{t=1}^T  \min(1, ( \breve{z}_t^{j_t})^\intercal \breve{\Sigma}_t  \breve{z}_t^{j_t}) \leq 2d \log \left( \frac{\TR(\breve{\Sigma}_1^{-1}) + \breve{M}_L^2 \sum_{t=1}^T (X_T^{j_t})^2 }{d} \right) \label{eq:breve_r2_5_bound}
\end{align}

Combining \eqref{eq:breve_r2_4_bound} and \eqref{eq:breve_r2_5_bound}, we can bound the second part of \eqref{eq:boundR2_2_breve} to the following
\begin{equation}\label{eq:breve_r2_6_bound}
 \EXP\Big[\sum_t (\breve{X}_T^i)^2||\breve{\Sigma}^{0.5}_t \breve{z}_t^i ||^2\Big] \leq \mathcal{O} \Big(  \EXP\Big[ (\breve{X}_T^i)^4 \log( \sum_{t=1}^T (\breve{X}_T^{j_t})^2) \Big] +  \EXP\Big[ (\breve{X}_T^i)^2 \log( \sum_{t=1}^T (\breve{X}_T^{j_t})^2) \Big]  \Big).
\end{equation}

The bound on $\breve{R}^i_2(T)$ in Lemma \ref{lem:breve_regret_terms} then follows by combining \eqref{eq:boundR2_1_breve}-\eqref{eq:breve_r2_6_bound} with the bound on $\EXP\Big[ (\breve{X}_T^i)^q \log( \sum_{t=1}^T (\breve{X}_T^{j_t})^2) \Big] $ for $q=2,4$  in Lemma \ref{lem:breve_x_bound} in the appendix.

\end{proof}

\begin{lemma}\label{lem:breve_x_bound}
For any $q\geq 1$, we have
\begin{align}
\EXP \left[ (\breve{X}_T^i)^q \log(\sum_t (\breve{X}_T^{j_t})^2 ) \right] &\leq  \left(\breve{\sigma} \right)^q \tilde{\mathcal{O}}(1)
\end{align}
\end{lemma}
\begin{proof}

\begin{align*}
\EXP \left[ (\breve{X}_T^i)^q \log(\sum_t (\breve{X}_T^{j_t})^2 ) \right] &\leq \EXP \left[ (\breve{X}_T^i)^q \log( \max(e, \sum_t (\breve{X}_T^{j_t})^2 ) ) \right] \\
& \leq \sqrt { \EXP [ (\breve{X}_T^i)^{2q}] ~~\EXP [\log^2 \big( \max(e, \sum_t \breve{X}_T^{j_t})^2 ) \big) ]} 
\end{align*}
where the second inequality follows from the Cauchy-Schwarz inequality. Now, $\log^2(x)$ is a concave function for $x \geq e$. Therefore, using Jensen's inequality we can write,
\begin{align*}
\EXP \log^2 \big( \max(e, \sum_t \breve{X}_T^{j_t})^2 ) \big) &\leq \log^2 ( \EXP \max(e, \sum_t \breve{X}_T^{j_t})^2 ) )\\
&\leq \log^2 \left( e + \EXP (\sum_t \breve{X}_T^{j_t})^2 ) \right) \\
&\leq \log^2 \left(e + T \mathcal{O}(\breve{\sigma}^2 ~\log T) \right) \\
&= \tilde{\mathcal{O}}(1)
\end{align*}

where we used Lemma 4 in the last inequality. Similarly, $\EXP [ (\breve{X}_T^i)^{2q}] \leq \left(\breve{\sigma} \right)^{2q} \mathcal{O}(\log T)$. Therefore, combining the above inequalities we have the following:
\begin{align*}
\EXP \left[ (\breve{X}_T^i)^q \log(\sum_t (\breve{X}_T^{j_t})^2 ) \right] &\leq \sqrt { \EXP [ (\breve{X}_T^i)^{2q} ]~~ \EXP\Big[ \log^2 \left( \max(e, \sum_t \breve{X}_T^{j_t})^2 ) \right) \Big] } \\
& \leq \breve{\sigma}^q \tilde{\mathcal{O}}(1)
\end{align*}

\end{proof}

\section{Simulation Details}

We consider homogeneous scalar system ($\abs{K}=1$) with $\MAT A=1,\MAT B=0.3,\MAT D=0.5,\MAT E=0.2, \MAT Q=1,\bar{\MAT Q}=1, \MAT R=1$,and $\bar{\MAT R}=0.5$. We set the local noise variance $\sigma_w^2=1$.  

For the regret plots in Figure \ref{fig:regreta},\ref{fig:regretb}, we set the common noise variance to $\sigma_v^2+\sigma_{v^0}^2=1$. 
The prior distribution used in the simulation are set according to \textbf{(A3)} and \textbf{(A4)} with $\breve \mu(\ell)=[1,1]$, $\bar \mu(\ell)=[1,1]$, $\bar{\Sigma}_1=I$, and $\breve{\Sigma}_1=I$, $\breve\Theta = \{\breve{\theta}: \MAT A + \MAT B  \breve {\MAT L}(\breve \theta) \leq \delta \}$, $\bar\Theta = \{\bar{\theta}: \MAT A + \MAT D + (\MAT B+\MAT E) \bar{\MAT L}(\bar \theta) \leq \delta \}$ and $\delta = 0.99$. 

In the comparison of {\tt TSDE-MF} method with {\tt TSDE} in Figure \ref{fig:regretc}, we consider the same dynamics and cost parameters as above but without common noise (i.e. $\sigma_v^2+\sigma_{v^0}^2=0$).  We set the prior distribution parameters to $\breve \mu(\ell)=[0,0]$, $\bar \mu(\ell)=[0,0]$, $\bar{\Sigma}_1=I$, and $\breve{\Sigma}_1=I$ and $\delta=2.3$ in the definition of $\bar{\Theta},\breve{\Theta}$. Note that even though $\delta=2.3$ does not satisfy \textbf{(A5)}, the results show that {\tt TSDE-MF} continues to have good performance in practice.

\section{Comparison with other agent selection schemes}
In \texttt{TSDE-MF}, we update the posterior probability $\breve p_t$ on
$\breve \theta$ using $\{ x^{i^*_s}_s, u^{i^*_s}_s, x^{i^*_s}_{s+1} \}_{1 \le
s < t}$, where $i^*_s = \arg \max_{i \in N} (\breve z^i_s)^\TRANS \breve
\Sigma_s \breve z^i_s$. This particular choice of the agent selection rule
implies that while deriving a bound on $\breve R^i_2(T)$, we can upper
bound%
\footnote{The precise argument is a bit more subtle; see proof of Lemma~\ref{lem:breve_regret_terms} for details.}
$\sum_{t=1}^T (\breve z^i_t)^\TRANS \breve \Sigma_t \breve z^i_t$ by
$\sum_{t=1}^T (\breve z^{i^*_t}_t)^\TRANS \breve \Sigma_t \breve z^{i^*_t}_t$.
This, in turn, allows us to bound the regret of $\breve R^i_2(T)$ in terms of
$\EXP[ \log^2 (\sum_{t=1}^T (X^{i^*_t}_T)^2)]$, which we show is $\tilde
O(1)$.

There are three other choices for the agent
selection rule: (i)~picking a
specific agent, or (ii)~picking an agent at random, or~(iii) using the entire
trajectory $(\VVEC {\breve x}_{1:t}, \VVEC {\breve u}_{1:t-1})$, where $\VVEC
{\breve x}_t = (\breve x^i_t)_{i \in N}$ and $\VVEC {\breve u}_t = (\breve
u^i_t)_{i \in N}$. 

\begin{figure*}[!ht]
    \centering
    \hfill
    \begin{subfigure}[t]{0.32\linewidth}
      \centering
      \includegraphics[width=\textwidth]{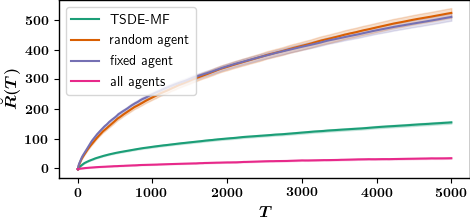}
      \caption{$\breve R(T)$ vs $T$}
    \end{subfigure}%
    \hfill
    \begin{subfigure}[t]{0.32\linewidth}
      \centering
      \includegraphics[width=\textwidth]{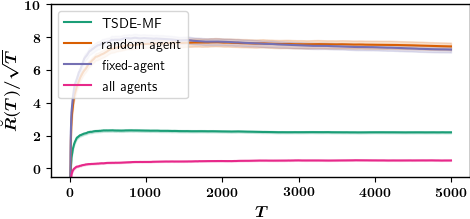}
      \caption{$\breve R(T)/\sqrt{T}$ vs $T$}
    \end{subfigure}%
    \hfill\null
    \caption{Impact of different agent selection schemes on expected regret.} 
    \label{fig:comparison}
\end{figure*}

If we follow approach~(i) and arbitrarily pick an agent, say $j$, and
update the posterior distribution $\breve p_t$ on $\breve \theta$ using
$(x^j_{1:t}, u^j_{1:t-1})$. This would mean that we can directly use the
result of \cite{ouyang2017control, ouyang2019posterior} to bound the regret of
$\breve R^j(T)$. However, we would still need to bound $\breve R^i(T)$ for $i
\neq j$. In this case, we can follow the argument similar to one presented in
the supplementary file to bound $\breve R^i_0(T)$ and $\breve R^i_1(T)$, but
the bound on $\breve R^i_2(T)$ does not work because we are not able to bound 
$\sum_{t=1}^T (\breve z^i_t)^\TRANS \breve \Sigma_t \breve z^i_t$ in terms of
an expression involving $X^j_T$. Similar limitations hold for
alternative~(ii).

We conducted numerical experiments to check if these alternatives perform better
in practice, which are presented in Fig.~\ref{fig:comparison}, where we show $\breve R(T) =
\frac{1}{n} \sum_{i=1}^n \breve R^i(T)$ for the system model analyzed in
Sec.~\ref{sec:simulation}.
For alternatives (i) and (ii), their regret orders appear to be bounded by
$\tilde{\mathcal{O}}(\sqrt{T})$, but they clearly perform worse than the
proposed method. For alternative (iii), the regret is slightly better than the
proposed method. However, implementing alternative (iii) requires complete
trajectory sharing among all agents. The extra computation and communication
cost of alternative (iii) could hinder its application to systems with a large
number of agents.

\end{document}